\documentclass[12pt]{article}

\usepackage{hyperref}
\usepackage{amssymb,amsmath,amsthm,amsfonts,amscd}
\usepackage{enumerate}
\usepackage[titletoc,title]{appendix}
\usepackage{geometry}
\usepackage{setspace}
\usepackage[latin2]{inputenc}
\usepackage{t1enc}
\usepackage[mathscr]{eucal}
\usepackage{indentfirst}
\usepackage{graphicx}
\usepackage{graphics}
\usepackage{pict2e}
\usepackage{epic}
\numberwithin{equation}{section}
\usepackage{epstopdf}

\newtheorem{theorem}{Theorem}

\newtheorem*{definition*}{Definition}
\newtheorem*{theorem*}{Theorem}
\newtheorem*{lemma*}{Lemma}
\newtheorem{lemma}[theorem]{Lemma}

\newtheorem{fact}[theorem]{Fact}
\newtheorem*{fact*}{Fact}
\newtheorem{conjecture}{Conjecture}
\newtheorem{corollary}[theorem]{Corollary}
\newtheorem*{remark*}{Remark}
\newtheorem{proposition}[theorem]{Proposition}

\newcommand{\Exp}{{\bf E}}
\newcommand{\Var}{{\bf Var}}
\newcommand{\Cov}{{\bf Cov}}
\newcommand{\Inf}{{\bf I}}
\newcommand{\Ent}{{\bf H}}
\newcommand{\Tr}{{\text{Tr}_n}}
\title{On the Fourier Entropy Influence Conjecture \\ for Extremal Classes}
\author{Guy Shalev%
\thanks{Department of Computer Science, Tel-Aviv University. \newline The research leading to these results has received funding from the Len Blavatnik amd the Blavatnik Family foundation.}
\\ guyshalev2@gmail.com}
\date{\today}

\begin{document}

\maketitle



\maketitle

\begin{abstract} The Fourier Entropy-Influence (FEI) Conjecture of Friedgut and Kalai \cite{FK96} states that $\Ent[f] \leq C \cdot \Inf[f]$ holds for every Boolean function $f$, where $\Ent[f]$ denotes the spectral entropy of $f$, $\Inf[f]$ is its total influence, and $C > 0$ is a universal constant. Despite significant interest in the conjecture it has only been shown to hold for some classes of Boolean functions such as symmetric functions and read-once formulas.

In this work, we prove the conjecture for the extremal cases, i.e. functions with small influence and functions with high entropy. Specifically, we show that:
\begin{itemize}
  \item FEI holds for the class of functions with $\Inf[f] \leq 2^{-cn}$ with the constant $C = 4 \cdot \frac{c+1}{c}$. Furthermore, proving FEI for a class of functions with $\Inf[f] \leq 2^{-s(n)}$ for some $s(n) = o(n)$ will imply FEI for the class of all Boolean functions.
  \item FEI holds for the class of functions with $\Ent[f] \geq cn$ with the constant $C = \frac{1 + c}{h^{-1}(c^2)}$. Furthermore, proving FEI for a class of functions with $\Ent[f] \geq s(n)$ for some $s(n) = o(n)$ will imply FEI for the class of all Boolean functions.
\end{itemize}

Additionally, we show that FEI holds for the class of functions with constant $\|\widehat{f}\|_1$, completing the results of \cite{CKLS16} that bounded the entropy of such functions. We also improve the result of \cite{WWW14} for read-k decision trees, from $\Ent[f] \leq O(k) \cdot \Inf[f]$ to $\Ent[f] \leq O(\sqrt{k}) \cdot \Inf[f]$. Finally, we suggest a direction for proving FEI for read-k DNFs, and prove the Fourier Min-Entropy/Influence (FMEI) Conjecture for regular read-k DNFs.
\end{abstract}

\clearpage

\section{Introduction} \label{Introduction}

Boolean functions $ f \colon \{-1,1\}^n \to \{-1,1\}$ are one of the most basic objects in the theory of computer science. The Fourier analysis of Boolean functions has become prominent over the years as a powerful tool in the study of Boolean functions, with applications in many fields such as complexity theory, learning theory, social choice, inapproximability, metric spaces, random graphs, coding theory, etc. For a comprehensive survey, see the book \cite{O'DBook}.

For Boolean-valued functions, by applying Parseval's identity we have $\sum_{S \subseteq [n]} \widehat{f}(S)^2 = 1$ and therefore the squared Fourier coefficients $ \widehat{f}(S)^2 $ can be viewed as a probability distribution $\mathcal{S}_f$, named the \textit{spectral distribution} of $f$.
The \textit{spectral entropy} of $f$ is defined to be the Shannon entropy of $\mathcal{S}_f$, namely 
$\Ent[f] = \sum_{S \subseteq [n]} \widehat{f}(S)^2 \log \frac{1}{\widehat{f}(S)^2}$. This can be intuitively thought as how "spread out" the Fourier coefficients of $f$ are.
The \textit{total influence} of a function $f$, one of the most basic measures of a Boolean function, can be defined as $\Inf[f] = \sum_{S \subseteq [n]} \widehat{f}(S)^2 |S| = \textbf{E}_{S \sim \mathcal{S}_f}[|S|]$, the expected size of a subset $S$ according to the spectral distribution, and can be intuitively thought of measuring the concentration of $f$ on "high" levels.

The Fourier Entropy Influence conjecture, posed by Friedgut and Kalai \cite{FK96} states that for any Boolean function the ratio of its spectral entropy and its total influence is upper-bounded by a universal constant.

\begin{conjecture} \label{FEI} (\cite{FK96}) \sloppy
There exists a universal constant $C > 0$  such that for all $f \colon \{-1,1\}^n \to \{-1,1\}$ with influence $\Inf[f]$ and spectral entropy $\Ent[f]$ we have $\Ent[f] \leq C \cdot \Inf[f]$.
\end{conjecture}

The original motivation for the conjecture in \cite{FK96} emerged from studying threshold
phenomena of monotone graph properties in random graphs. Specifically for a function $f$ that represents a monotone property of a graph with $n$ vertices (e.g. connectivity), FEI implies that $\Inf[f] \geq c \log ^2 n $. The best known bound as of today due to Bourgain and Kalai \cite{BK97} is $\Inf[f] \geq c \log ^{2-\varepsilon} n $ for every $\varepsilon > 0$.

Proving Conjecture \ref{FEI} will have other interesting applications. Probably the most important consequence of the conjecture is its implication of a variant of Mansour's conjecture from 1995 \cite{Man95} stating that if a Boolean function can be represented by a DNF formula with $m$ terms, then most of its Fourier weight is concentrated on a set of coefficients of size at most $\textit{poly}(m)$. Combined with results by Gopalan et al. \cite{GKK08a} this in turn will result in an efficient learning algorithm for such DNFs in the agnostic model, a central open problem in computational learning theory. Furthermore, sufficiently strong versions of Mansour's Conjecture would yield improved pseudorandom generators for DNF formulas. See \cite{Kal07}, \cite{OWZ11} for more details on this implication.

FEI is also closely related to the fundamental KKL theorem \cite{KKL88} stating that for every Boolean function, $\max_{i \in [n]}\Inf_{i}[f] \geq \Var[f] \cdot \Omega( \frac{\log n}{n})$. We define $\Ent_{\infty}[f] = \min_{S} \lbrace \log \frac{1}{\widehat{f}(S)^2} \rbrace $, the \textit{min-entropy} of $f$. It is easy to verify that $\Ent[f] \geq  \Ent_{\infty}[f]$. A natural relaxation of FEI is the following weaker Fourier Min-Entropy Influence conjecture:

\begin{conjecture}(FMEI)
\sloppy
There exists some $C > 0$ such that for any $ f \colon \{-1,1\}^n \to \{-1,1\}$  we have $\Ent_{\infty}[f] \leq C \cdot \Inf[f]$.
\end{conjecture}

KKL can be directly derived from FMEI (and therefore is clearly implied by FEI). In the other direction, one can easily prove FMEI for monotone functions using KKL (see \cite{OWZ11}). We note that FEI for monotone functions is still an open problem.

\subsection{Prior Work}
Despite many years of attention, Conjecture \ref{FEI} remains open, but some significant steps towards proving it have been made. For example, a weaker folklore version of FEI, where instead of a universal constant $C$ we settle for a $\log n$ factor, is known to be true even for the more general case of real-valued Boolean functions.

\begin{lemma}[Weak FEI] \label{Weak FEI}
Let $ f \colon \{-1,1\}^n \to \mathbb{R}$ be some function with
 $\|\widehat{f}\|_2 = 1$. Then $\Ent[f] \leq \log (n + 1) \cdot (\Inf[f] + 1)$.
\end{lemma}

This can be proved in several different ways, as done in \cite{KMS11},\cite{OWZ11} and \cite{WWW14}. It should be noted that $O(\log n)$ is indeed tight for non-Boolean functions, so proofs of FEI will have to make use of the fact that $f$ is Boolean-valued. The tightness can be seen by the following example:
\[ f(x) = \frac{x_1 + x_2 + ... + x_n}{\sqrt{n}} \]
and it is easy to verify that $\|\widehat{f}\|_2 = 1$, and also that $\Inf[f] = 1$ and $\Ent[f] = \log n$. For Boolean-valued functions this $\log n$ bound has been recently improved by Gopalan \textit{et al.} in \cite{GSTW16} to $\log s[f]$, where $s[f]$ is \textit{max sensitivity} of the function: the sensitivity of $x \in \{-1,1\}^n$ in a function $f$, denoted $s(f,x)$, is the number of indices $i \in [n]$ for which $f(x) \neq f(x^{ \oplus i})$, and the max sensitivity of $f$ is defined as $s[f] := \max_{ x \in \{-1,1\}^n } s(f,x)$. Clearly, for all functions $s[f] \leq n$.

Furthermore, FEI has been verified for several families of Boolean functions. O'Donnel, Wright and Zhou \cite{OWZ11} proved it for symmetric functions by using the fact that derivatives of symmetric functions are very noise sensitive. They also prove FEI for the class of read-once decision trees.

In another paper Das, Pal and Visavaliya \cite{DPV11} show that FEI holds with universal constant $2+\varepsilon$ for a random function, as $\Inf[f]$ is strongly concentrated around its mean $\frac{n}{2}$, and the spectral entropy of a function is always bounded by $n$. We give another proof of this fact (with a worse constant), by proving FEI for functions with entropy linear in $n$, as is the case for random functions.

In \cite{KMS11}, Keller, Mossel and Schlank generalize FEI to the biased setting. Furthermore, for functions with almost all of their Fourier weight on the lowest $k$ levels, they upper-bound the spectral entropy by $O(k)$.

In the paper \cite{OT13}, O'Donnell and Tan study FEI under composition: given functions $F: \{-1, 1\}^{k} \to \{-1, 1\}$ and
$g_1, ... , g_k: \{-1, 1\}^l \to \{-1, 1\}$, they ask what properties do $F$ and $g_i$ must satisfy for the FEI conjecture to hold for the disjoint composition $f(x^1, . . . , x^k) = F(g_1(x^1), . . . , g_k(x^k))$? To make progress they present a strengthening of FEI which they call FEI$^{+}$ - a generalization of FEI to product distributions. They prove that FEI$^{+}$ composes, in the sense that if $F$ and $g_i$ respect FEI$^{+}$ with factor $C$, then so does their composition. They also prove FEI$^{+}$ with factor $C = O(2^k)$, where $k$ is the arity of $f$ (instead of $C$ being a constant). Together with their main result, this is enough to prove FEI for read-once formulas.

In \cite{CKLS16}, Chakraborty \textit{et al.} prove a relaxation of FEI, bounding the spectral entropy with higher moments of $|S|$, where the original conjecture needs this bound to include only the first moment of $|S|$, namely $\Inf[f]$. They also prove FEI for read-once formulas with a more elementary method than the one of O'Donnell and Tan.

Independently, \cite{CKLS16} also give upper bounds on the entropy of a Boolean function in terms of several complexity measures - to name a few, they show that $\Ent[f] \leq O(\log \|\widehat{f}\|_1)$, and also that $\Ent[f] \leq O(\bar{d})$ where $\bar{d}$ is the average depth of a decision tree computing $f$. This implies FEI for the class $\left\lbrace f \colon \{-1,1\}^n \to \{-1,1\} : \|\widehat{f}\|_1 \leq L, \Inf[f] \geq c \right\rbrace$ where $c >0$ and $L>0$ are some constants.

This raises the natural question, whether the $\Inf[f] \geq c$ requirement is actually necessary or merely an artifact of the proof. For the $\|\widehat{f}\|_1$ complexity measure and other measures strongly related to it, we manage to overcome this condition by making subtle changes to the proof technique of \cite{CKLS16}, generalizing the bound and thus proving FEI for the class of functions with constant $\|\widehat{f}\|_1$. Another measure they use to bound the entropy, is the average depth of a decision tree computing $f$ (they show $\Ent[f] \leq O(\bar{d})$). For this measure, the $\Inf[f] \geq c$ requirement seems critical, as will be explained in the next paragraph.

In \cite{WWW14}, Wan, Wright and Wu present a new perspective of FEI as a communication (or rather, compression) game: one player randomly samples a set $S$ according to the distribution $\mathcal{S}_f$, and wishes to send it to another player using a short representation. The price of the protocol is the expected number of bits in the representation of $S \sim \mathcal{S}_f$. For a function $f$, we know from Shannon that the price of the protocol is lower bounded by the spectral entropy, so we are merely left with the challenge of finding a protocol for $f$ with expected price less than $O(\Inf[f])$. They formalize this into the following lemma:

\begin{lemma*}
Let $\mathcal{X} \sim \widehat{f}^2$, and let $P : 2^{[n]} \to \Sigma^*$ be a prefix-free protocol on alphabet $\Sigma$, except it outputs an empty string on the input $\emptyset$.
Then $\Ent[f] \leq \log \Sigma \cdot \Exp[P(\mathcal{X})] + 2 \cdot \Inf[f]$.

\end{lemma*}

They use this technique combined with observations regarding the covariance of decision trees to prove a theorem (that is also known due to \cite{CKLS16}) - that FEI holds for the class of functions $f$ computed by decision trees with constant average depth and $\Inf[f] \geq 1$. \cite{WWW14} also provide a reduction, showing that removing the requirement $\Inf[f] \geq 1$ from the latter theorem, would in fact result in proving FEI for all Boolean functions with $ \Inf[f] \geq \log n$. This gives more motivation to examine FEI for functions with low influence.

Using their protocol technique, \cite{WWW14} also achieve $\Ent[f] \leq O(k) \cdot \Inf[f]$ for read-k decision trees, thus proving FEI for read-k decision trees where $k$ is constant. They explicitly conjecture that the correct coefficient is actually $O(\log k)$ and provide a matching example. We improve their bound to $\Ent[f] \leq O(\sqrt{k}) \cdot \Inf[f]$, but share their belief that the correct bound could be $O(\log k) \cdot \Inf[f]$.

In \cite{Hod17}, Hod improves the lower bound on the conjectured universal constant for FEI to $C > 6.45$ via lexicographic functions, using composition techniques and biased Fourier analysis.

\subsection{Our Results}
Intrigued by the implicit and explicit difficulties of FEI for low influence functions, we prove FEI for functions with extremely low influence:

\begin{theorem}
\label{FEI low influence}
Let $c > 0$ be some constant. Let $ f \colon \{-1,1\}^n \to \{-1,1\}$ with $\Inf[f] \leq 2^{-cn}$ . Then $\Ent[f] \leq 4 \cdot \frac{c+1}{c} \cdot {\Inf[f]} $.
\end{theorem}

This result may seem at first somewhat disappointing, as interesting functions usually don't have such small total influence. Can we do better than this bound? Apparently not, at least without proving the full conjecture. Using a construction presented in \cite{WWW14} we show that any improvement of the last theorem will result in proving FEI:

\begin{theorem}
\label{FEI low influence is tight}
Let $s: \mathbb{N} \to \mathbb{R}$ such that $s(n) = o(n)$. Suppose that FEI holds for all $ f \colon \{-1,1\}^n \to \{-1,1\}$ with $\Inf[f] \leq 2^{-s(n)}$. Then FEI holds for all Boolean functions.
\end{theorem}
For example, proving FEI for the class of functions with $\Inf[f] \leq 2^{-\frac{n}{\log n}}$ will be enough to confirm Conjecture \ref{FEI}.

This result for functions with extremely low influence raises the question of the opposite extremal case - where the entropy is high, say, $cn$ for some $c \in (0,1)$. We provide analogous results for this extremal case.

\begin{theorem}
\label{FEI high entropy}
Let $c > 0$ be some constant. For any $ f \colon \{-1,1\}^n \to \{-1,1\}$ with $\Ent[f] \geq cn$ we have $\Ent[f] \leq \frac{1 + c}{h^{-1}(c^2)} \cdot \Inf[f]$, where $h^{-1}$ is the inverse of the binary entropy function.
\end{theorem}

\begin{theorem}
\label{FEI linear entropy is tight}
Let $s: \mathbb{N} \to \mathbb{R}$ such that $s(n) = o(n)$. Suppose that FEI holds for all $ f \colon \{-1,1\}^n \to \{-1,1\}$ with $\Ent[f] > s(n)$. Then FEI holds for all Boolean functions.
\end{theorem}

For example, proving FEI for the class of functions with $\Ent[f]  \geq \frac{n}{\log n}$ will confirm Conjecture \ref{FEI}. We also note that the other two extremal cases are easy, namely functions with exponentially low entropy and functions with total influence linear in $n$.

Independently from our work on the extremal classes, we also provide some improvements on previously known results. First, we modify the $\Ent[f] \leq O(\log \|\widehat{f}\|_1)$ bound of \cite{CKLS16} to include the influence and variance of the function, thereby showing that FEI holds for the class of functions with constant $\|\widehat{f}\|_1$, $\left\lbrace f \colon \{-1,1\}^n \to \{-1,1\} : \|\widehat{f}\|_1 \leq L \right\rbrace$.

\begin{theorem}
Let $ f \colon \{-1,1\}^n \to \{-1,1\}$ with $\|\widehat{f}\|_1 = L$. Then
$\Ent[f] \leq (4 \log L + 11) \Var(f) + 10 \cdot \Inf[f] $. In particular, from the edge isoperimetric inequality, we have $\Ent[f] \leq (4 \log L + 21) \cdot \Inf[f] $.
\end{theorem}

As a direct corollary, we can deduce FEI for functions with some related complexity measures that are constant. We note that some of these results have been previously known.
\begin{corollary}
FEI holds for functions with constant $\|\widehat{f}\|_1$, constant sub-cube partition, constant degree, constant decision tree depth, constant decision tree size, constant granularity or constant sparsity.
\end{corollary}

We also build and improve on the work of \cite{WWW14}. Inspired by their methods, we provide a hopefully promising direction towards proving FEI for read-k DNFs. We give an explicit protocol for the Tribes function which is a read-once DNF, and conjecture its possible generalization to a protocol for read-k DNFs, as a step towards read-k formulas, breaking the barrier of ``read-once'' as an assumption required for many of the results known today. As a first step, we prove FMEI for regular read-k DNFs, where by regular we mean (informally) that all clauses are more or less of the same width, and the number of clauses is exponential in that width.

\begin{theorem} 
Let $f$ be a regular read-k DNF, then $\Ent_{\infty}[f] = O(\Inf[f])$.
\end{theorem}

We also improve the result of \cite{WWW14} for read-k decision trees. \cite{WWW14} define the \textit{tree covariance} of a decision tree recursively as: $\Cov[T] = \Cov(g, h) + \frac{1}{2}(\Cov[T_0] + \Cov[T_1])$, where $g,h$ represent the functions defined by the left and right children of the root of $T$. They come up with a protocol for decision trees with price $4 \cdot \Inf[f] + 2 \cdot \Cov[T]$. Therefore, by proving $\Cov[T] \leq k \cdot \Var[f] \leq k \cdot \Inf[f]$ they obtain FEI for read-k decision trees with constant $k$. By improving the bound on the covariance to $\Cov[T] \leq O(\sqrt{k}) \cdot \Inf[f]$, we manage to also improve the constant achieved for FEI regarding this class. 

\begin{theorem} 
Let $ f $ be computed by a read-k decision tree. Then $\Cov[T] \leq O(\sqrt{k}) \cdot \Inf[f]$. As a result, FEI holds for read-k decision trees with constant $C = O(\sqrt{k})$.
\end{theorem}
We believe the tree covariance of a decision tree and its connection to other measures of the function it computes such as its variance and influence, might be of independent interest in the study of decision trees.

Finally, as an independent result, we refine the known connection between the size of a decision tree, and the spectral norm ($\|\widehat{f}\|_1$) of the function it computes. It is a well known fact that $\|\widehat{f}\|_1 \leq size(T)$, the size of a decision tree being the number of nodes in it. Our improvement involves the covariance of the nodes in the decision tree, and is given by the following lemma:
\begin{proposition}
For a Boolean function $f$ that is computed by a decision tree $T$:
\[ \|\widehat{f}\|_1 \leq \text{boundary\_size}(T) - \sum_{v \in inner(T)}{| \Cov(g_v, h_v) | } \]
\end{proposition}
Where $boundary\_size(T)$ is the number of nodes that have at least one child that is a leaf. The sum of covariances is over all inner nodes of $T$, i.e. nodes that have two non-leaf children. This improved bound is tight in some cases where the bound $\|\widehat{f}\|_1 \leq size(T)$ is far from it - for example, the parity function on $n$ variables with the natural tree that computes it.

\section{Preliminaries}

\subsection{Fourier Analysis of Boolean Functions}
It is well known that functions $ f \colon \{-1,1\}^n \to \mathbb{R}$ can be uniquely expressed as multi-linear polynomials:
\[ f =  \sum_{S \subseteq [n]} \widehat{f}(S) \chi_S(x) \]
where $\chi_S(x) = \prod_{i \in S} x_i$. This is known as the \textit{Fourier expansion} of $f$, and $\widehat{f}(S)$ are the \textit{Fourier coefficients} of the function.
For Boolean-valued functions $ f \colon \{-1,1\}^n \to \{-1,1\}$ Parseval's identity implies that
$\sum_{S \subseteq [n]} \widehat{f}(S)^2 = 1$, and therefore $ \{ \widehat{f}(S)^2 \}_{S \subseteq [n]} $ can be viewed as a probability distribution, named the \textit{spectral distribution} of $f$ and denoted $\mathcal{S}_f$.  
Two of the central complexity measures of a Boolean function can be defined using its spectral distribution:
 
\begin{definition*} The \textit{spectral entropy} of a function $ f \colon \{-1,1\}^n \to \{-1,1\}$ is the Shannon-entropy of the squared Fourier coefficients, namely
\[ \Ent[f] = \Exp_{S \sim \mathcal{S}_f} \left[ \log_{2} \frac{1}{\widehat{f}(S)^2} \right] =
\sum_{S \subseteq [n]} \widehat{f}(S)^2 \log \frac{1}{\widehat{f}(S)^2} \]
\end{definition*}

\begin{definition*} The \textit{influence} of a function $ f \colon \{-1,1\}^n \to \{-1,1\}$  (sometimes referred to as its \textit{total influence}) is
\[ \Inf[f] = \Exp_{S \sim \mathcal{S}_f} [|S|] =
\sum_{S \subseteq [n]} \widehat{f}(S)^2 |S| \]
\end{definition*}

The influence of a Boolean function also has a nice combinatorial interpretation. For $i \in [n]$, the influence of a variable $x_i$ in $f$ is $\textbf{Pr}[f(x) \neq f(x^{ \oplus i}) ]$, namely the probability that for a uniformly random input flipping the $i$'th bit will affect the result. An equivalent definition for the total influence of a function is $\Inf[f] = \sum_{i=1}^{n}\Inf_i[f]$.

It is sometimes useful to classify the Fourier coefficients by their level, where the level of $S$ is $|S|$. The weight of $f$ at level $k$ is denoted $W^{k}[f] = \sum_{|S| = k} \widehat{f}(S)^2$. Note that
$ \Inf[f] = \sum_{k=0}^{n} W^{k}[f] \cdot k $. Additionally, we use the following notations: $W^{\leq k}[f] = \sum_{|S| \leq k} \widehat{f}(S)^2$, and $W^{\geq k}[f] = \sum_{|S| \geq k} \widehat{f}(S)^2$.

We also use the decision tree model of computation, see \cite{O'DBook} for a formal definition. Given a tree $T$, we call the sub-tree corresponding to the $+1$ edge leaving the root the left sub-tree ($T_0$), and call the sub-tree corresponding to the $-1$ edge leaving the root the right sub-tree ($T_1$), and denote by $g$ and $h$ the corresponding functions to each sub-tree. We assume that no variable appears more than once in any root-to-leaf path of T (or else the tree can be easily simplified). For a node $v$ in $T$ we denote by $d(v)$ the depth of $v$ - its distance from the root node. We say that $T$ is a read-k decision tree if no variable is queried at more than $k$ nodes of $T$.

Given two functions $g, h \colon \{-1,1\}^n \to \mathbb{R}$ define $\Cov(g, h) = \Exp[(g(x) - \Exp[g])\cdot (h(x) - \Exp[h])]$. Following the definitions of \cite{WWW14}, we define the \textit{covariance of a decision tree}  $T$:
for an internal node $v$, let $g$ be the function computed by $v$'s left sub-tree and $h$ be the function computed by $v$'s right sub-tree. Then, define:
\begin{itemize}
 \item $\Cov[v] = \Cov(g, h)$
 \item $\Cov[T] = \sum_{v \in T} \Cov[v] \cdot 2^{-d(v)}$
\end{itemize}

Note that $\Cov[T]$ can be equivalently defined recursively as $\Cov[T] = \Cov(g, h) + \frac{1}{2} (\Cov[T_0] +\Cov[T_1])$, with the base case that $\Cov[T] = 0$ if $T$ has depth 0.

A DNF over Boolean variables $x_1, x_2, ..., x_n$ is the logical OR of terms, $T_1 \lor T_2 \lor ... \lor T_s$ each of which is a logical AND of literals \{$x_i$, $\bar{x_i}$\}. The number of literals in a term is called its width (sometimes we refer to it as the size of the term). A DNF is read-k if no variable appears in more than $k$ terms.

The Tribes function with width $w \in \mathbb{N}^+$ and $s$ tribes, is a read-once DNF on $n = s \cdot w$ variables, where all terms are of width exactly $w$:
\[ \text{Tr}_{w,s} (x_1,x_2, ..., x_{sw}) = (x_1 \land ... x_w) \lor ... \lor (x_{(s-1)w + 1} \land ... \land x_{sw}) \]
For $w \in \mathbb{N}^+$, we choose $s$ to be the largest integer such that $1 - (1-2^{-w})^s \leq \frac{1}{2}$, so $f$ will be as unbiased as possible. Then we denote by $\Tr$, defined only for such pairs of $w,s$, to be the (essentially) unbiased Tribes function on $n$ variables. Due to Proposition 4.12 in \cite{O'DBook}, $s \thickapprox \ln(2)2^w$, $n \thickapprox \ln(2) w 2^w$.

Finally, we present the definition of regular DNFs:
\begin{definition*} Let $C_1,C_2 > 0$. We say $f = T_1 \lor T_2 \lor ... \lor T_s$ is a $(C_1, C_2)$-regular DNF (or just ``regular''), if there exists some $w \in \mathbb{N}$ s.t. the number of variables in each clause respects $C_1 w \leq \text{size}(T_i) \leq w$, and the number of clauses is $s = 2^{C_2 w}$.
\end{definition*}

Expanding on the mentioned notions of the Shannon entropy $\Ent[f]$ and the min-entropy $\Ent_\infty[f]$, the \textit{Renyi entropy} of a distribution $\mathcal{X}$ (we discuss only $\mathcal{S}_f$) is defines as follows:

\[ \Ent_a[\mathcal{X}] = \frac{1}{1-a} \log \left( \sum_{i=1}^{n} {p_i}^{a}\right). \]
Where $p_i$ are the probabilities of possible instances in $\mathcal{X}$ - in our case, these are the squared Fourier coefficients.

It can be seen that for $a \to 1$, the Renyi entropy converges to the Shannon entropy, therefore we denote $\Ent_1[f] = \Ent[f]$. Furthermore, when $a \to \infty$, the Renyi entropy converges to the min-entropy $\Ent_\infty$. It is known that for a fixed distribution, the function $\Ent_a[\mathcal{X}]$ is non-increasing in $a$.

\subsection{Edge Isoperimetric Inequality}
The simplest form of the Edge Isoperimetric Inequality states that for any Boolean function $f$, $\Var(f) \leq \Inf[f]$. 
We also rely on the following edge isoperimetric inequality, see e.g.\ Theorem 2.39 in \cite{O'DBook}:
\begin{fact} \label{Edge Isoperimetric Inequality 1} 
Let $f$ be a Boolean function. Denote $\alpha = \min ({Pr[f=1]}, {Pr[f=-1]})$, then
$2 \alpha \log \frac{1}{\alpha} \leq \Inf[f] $.
\end{fact}
Keeping the notation $ \alpha = \min(Pr[f=1],Pr[f=-1]) $, it is easy to see that $\alpha$ can be ``replaced'' by the variance, losing only a constant multiplicative factor:
\[\Var(f) = \Exp(f^2) - \Exp(f)^2 = 1 - (1 - 2 \alpha)^2 = 4\alpha - 4\alpha ^2 = 4 \alpha (1 - \alpha)\]
Since $\frac{1}{2} \leq 1 - \alpha \leq 1$, we have:
\[2\alpha \leq \Var(f) \leq 4\alpha \]

\begin{lemma} \label{Edge Isoperimetric Inequality 2} 
Let $f$ be a Boolean function, then $\frac{1}{2} \cdot \Var(f) \cdot \log \frac{1}{\Var(f)} \leq \Inf[f] $.
\end{lemma}

\begin{proof}
\[ \Inf[f] \geq 2 \alpha \log \frac{1}{\alpha}
\geq \frac{1}{2} \Var(f) \cdot \log \frac{2}{\Var(f)}
 \geq  \frac{1}{2} \Var(f) \cdot \log \frac{1}{\Var(f)} \]
Where the second inequality is due to  $2\alpha \leq \Var(f) \leq 4 \alpha$.
\end{proof}

\begin{lemma} \label{Edge Isoperimetric Inequality 3} 
Let $f$ be a Boolean function  with $\Inf[f] < 1$, then $ \Var(f) \leq 2 \cdot \frac{\Inf[f]}{\log \frac{1}{\Inf[f]}} $.
\end{lemma}

\begin{proof}
The requirement that $\Inf[f] < 1$ is necessary, or else the term $\log \frac{1}{\Inf[f]}$ is non-positive.
We derive the new inequality from the proof of Lemma \ref{Edge Isoperimetric Inequality 2}:
\[ \frac{\Inf[f]}{\log \frac{1}{\Inf[f]}}
\geq \frac{\frac{1}{2} \Var(f) \cdot \log \frac{2}{\Var(f)}}{\log \left( \frac{2}{\Var(f) \cdot \log \frac{2}{\Var(f)}}\right)}
\geq \frac{1}{2} \Var(f) \cdot \frac{ \log \frac{2}{\Var(f)}}{\log \frac{2}{\Var(f)} - \log \log \frac{2}{\Var(f)}}
\geq \frac{1}{2} \Var(f)
\]

\end{proof}

\subsection{Tensorization of FEI}
Let $ f \colon \{-1,1\}^n \to \{-1,1\}$, $g \colon \{-1,1\}^m \to \{-1,1\}$ be two Boolean functions.
Define $h = f \otimes g$ to be their \textit{tensor product}: $ h \colon \{-1,1\}^{n+m} \to \{-1,1\}$, and 
$h(\textbf{x}, \textbf{y}) = f(\textbf{x}) \cdot g(\textbf{y})$. 

In \cite{Kal07} it has been noted that FEI tensorizes in the following sense:
\begin{fact} \label{I H and N Tensorize}
For Boolean functions $f, g$ and $h = f \otimes g$:
\begin{itemize}
\item $\Inf[h] = \Inf[f] + \Inf[g]$
\item $\Ent[h] = \Ent[f] + \Ent[g]$
\item $n_h = n_f + n_g$, where $n_f$ denotes the number of variables of the function $f$.
\end{itemize}
\end{fact}

We call $h = f \otimes f$ the self-tensorization of $f$. As stated in the following lemma, the tensorization technique allows us to deduce FEI for a class closed under self-tensorization (i.e., for all $f \in \mathcal{C}$ we have $f \otimes f \in \mathcal{C}$)
by proving FEI for that class with a sub-linear additive term. Obviously, the class of all Boolean functions is closed under self-tensorization.
\begin{lemma}
Suppose we have a class of Boolean functions $\mathcal{H}$ that is closed under self-tensorization, a constant $C > 0$ and some function $s(n) = o(n)$. If $\Ent[f] \leq C \cdot \Inf[f] + s(n)$ for all $f \in \mathcal{H}$, then $\Ent[f] \leq C \cdot \Inf[f]$ for all $f \in \mathcal{H}$.
\end{lemma}

\begin{proof}
Let $f \in \mathcal{H}$ be a Boolean function on $n$ variables. Define $f_1 = f$, and $f_{i+1} = f_i \otimes f_i$.
$f_k \in \mathcal{H}$ since $\mathcal{H}$ is closed under self-tensorization, and therefore we have 
\[ \Ent[f_k] \leq C \cdot \Inf[f_k] + s(n_{f_k}) \]
By Fact \ref{I H and N Tensorize}:
\[ 2^k \cdot \Ent[f] = \Ent[f_k] \leq C \cdot \Inf[f_k] + s(n_{f_k}) = C \cdot 2^k \cdot \Inf[f] + s(2^k \cdot n) \]
Dividing by $2^k$, we get:
\[ \Ent[f] \leq C \cdot \Inf[f] + \frac{s(2^k \cdot n)}{2^k} \]
Fixing $n$ and taking $k$ to infinity, we get $\frac{s(2^k \cdot n)}{2^k} \xrightarrow[k \rightarrow \infty]{} 0$, and therefore  $\Ent[f] \leq C \cdot \Inf[f]$.
\end{proof}

Additionally, we note that the min-entropy tensorizes as well: for $f,g$ and $h = f \otimes g$ as stated above, $\Ent_{\infty}[h] = \Ent_{\infty}[f] + \Ent_{\infty}[g]$, so a similar proof will suffice for an analogous result.

\begin{lemma} \label{FMEI tenzorizes}
Suppose we have a class of Boolean functions $\mathcal{H}$ that is closed under self-tensorization, a constant $C > 0$ and some function $s(n) = o(n)$. If $\Ent_{\infty}[f] \leq C \cdot \Inf[f] + s(n)$ for all $f \in \mathcal{H}$, then $\Ent_{\infty}[f] \leq C \cdot \Inf[f]$ for all $f \in \mathcal{H}$.
\end{lemma}

\section{FEI for Low Influence Functions}

In this section we prove FEI for the class of functions with exponentially low influence (in $n$), and then show that improving this will imply Conjecture \ref{FEI}. To state this formally, we introduce some notations and consider the following classes of functions:
\begin{itemize}
	\item The class of Boolean functions on $n$ variables $\mathbf{BF}_n = \lbrace f \colon \{-1,1\}^n \to \{-1,1\} \rbrace $, and the class of all Boolean functions $\mathbf{BF} = \bigcup\limits_{n=1}^{\infty} \mathbf{BF}_n$.
	\item The class of functions with exponentially-low influence. For every constant $c > 0$, define $\mathbf{ELI}_c = \bigcup\limits_{n=1}^{\infty} \left\lbrace f \in \mathbf{BF}_n : \Inf[f] \leq 2^{-cn} \right\rbrace$.
	\item The class of functions with ``almost'' exponentially-low influence. For every function $s: \mathbb{N} \to \mathbb{R}$ such that $s(n) = o(n)$, define $\mathbf{AELI}_s = \bigcup\limits_{n=1}^{\infty} \left\lbrace f \in \mathbf{BF}_n : \Inf[f] \leq 2^{-s(n)} \right\rbrace$.
	\item The class of functions with influence larger than $1$, $\mathbf{IL1} = \bigcup\limits_{n=1}^{\infty} \left\lbrace f \in \mathbf{BF}_n : \Inf[f] \geq 1 \right\rbrace$.
\end{itemize}

Formally, we show that for any $c \in (0,1) $ FEI holds for the class $\mathbf{ELI}_c$ with constant $C = O(\frac{1}{c})$. We then show that improving on this result by proving FEI for any $\mathbf{AELI}_s$ will actually imply FEI for the class $\mathbf{IL1}$. As a simple corollary of Theorem \ref{FEI linear entropy is tight} that we will later prove, this will imply FEI for $\mathbf{BF}$, i.e.\ Conjecture \ref{FEI}.

\subsection{Proving FEI for ELI}

We restate and prove Theorem \ref{FEI low influence} as follows:
\begin{theorem}\label{FEI-ELI}
For all $f \in \mathrm{ELI}_c$, $\Ent[f] \leq 4 \cdot \frac{c+1}{c} \cdot {\Inf[f]} $.
\end{theorem}

\begin{proof}
Let $f \in \mathrm{ELI}_c$, and denote $\Inf[f] = 2^{-c'n} \leq 2^{-cn}$, where $c' \geq c$.
We use the concentration method presented by \cite{CKLS16} to show 
$\Ent[f] \leq 4 \cdot \frac{c+1}{c} \cdot {\Inf[f]}$. We partition the Fourier coefficients to a family $\mathscr{F} = \{\widehat{f}(\emptyset)^2 \}$, and its complement $\mathscr{F}^C$. These families have Fourier weight of $1-\Var(f)$ and $\Var(f)$ respectively. By a known formula of entropy partition, we have:

\[ \Ent[f] = \Var(f) \cdot \Ent[\mathscr{F}^C] + (1 - \Var(f)) \cdot \Ent[\mathscr{F}] + h(\Var(f)) \]
where the entropies of the families are of the adequately normalized distributions and $h(p) := p \log(\frac{1}{p}) + (1-p) \log(\frac{1}{1-p})$ is the binary entropy function.
Note that $\Ent[\mathscr{F}] = 0$, since $\mathscr{F}$ contains only one element.
Also note that $\Ent[\mathscr{F}^C] < n$, as $| \mathscr{F}^C | < 2^n$. Therefore we have:

\begin{equation} 
 \label{eq_h_to_var} \Ent[f] \leq \Var(f) \cdot n + h(\Var(f))
\end{equation}

If $\Inf[f] \geq \frac{1}{2}$ then by our assumption it follows that $cn \leq 1$. From this and a from Weak FEI (Lemma \ref{Weak FEI}) we get $\Ent[f] \leq n \cdot \Inf[f] \leq \frac{1}{c} \cdot \Inf[f]$ so we are done.
Otherwise, we can assume $\Inf[f] < \frac{1}{2}$. To bound the second term of inequality \ref{eq_h_to_var}, $h(\Var(f))$, we note that for $p \leq \frac{1}{2}$, we have $p \log \frac{1}{p} \geq (1-p) \log \frac{1}{1-p}$, so $h(p) \leq 2 p \log \frac{1}{p} $. Acknowledging the fact that $\Var(f) \leq \Inf[f] \leq \frac{1}{2}$ and applying Lemma \ref{Edge Isoperimetric Inequality 2}, we have:

\begin{equation} 
 \label{eq_h_to_var_second} h(\Var(f)) \leq 2 \cdot \Var(f) \cdot \log \frac{1}{\Var(f)} \leq 4 \cdot \Inf[f]
\end{equation}

We can bound the first term of inequality \ref{eq_h_to_var} by applying Lemma \ref{Edge Isoperimetric Inequality 3}:
\begin{equation} 
 \label{eq_h_to_var_first} \Var(f) \cdot n \leq
 2 \cdot \frac{\Inf[f]}{\log \frac{1}{\Inf[f]}} \cdot n \leq 2 \cdot \frac{\Inf[f]}{c'n} \cdot n = 
\frac{2}{c'} \cdot \Inf[f] \leq \frac{2}{c} \cdot \Inf[f]
\end{equation}

Inserting \ref{eq_h_to_var_second}, \ref{eq_h_to_var_first} into \ref{eq_h_to_var} we obtain the wanted result:
\[ \Ent[f] \leq \frac{2}{c} \cdot \Inf[f] + 4 \cdot \Inf[f] = 4 \cdot \frac{c+1}{c} \cdot \Inf[f] \]
\end{proof}

There is also an alternative proof for theorem \ref{FEI-ELI} using the protocol method of \cite{WWW14}: Intuitively, consider the following trivial protocol: if the sampled $S$ is non-empty, send $n$ bits, where the $i$'th bit is set to $1$ if $i \in S$ and otherwise is $0$. If $S=\emptyset$ is sampled, the protocol sends the empty string. By Lemma \ref{Edge Isoperimetric Inequality 3} we have $\Var(f) \leq 2 \frac{\Inf[f]}{cn}$, so the average cost of this protocol will be: $ \Var(f) \cdot n = \frac{2}{c} \cdot \Inf[f]$ which also gives us FEI for $\mathrm{ELI}_c$ with constant $\Theta(\frac{1}{c})$ .


\subsection{Proving FEI for AELI Implies FEI Completely}

\begin{lemma}
\label{AELI implies IL1}
Let $s: \mathbb{N} \to \mathbb{R}$ such that $s(n) = o(n)$. Suppose that FEI holds for some class $\mathbf{AELI}_s$ with universal constant $C$, Then FEI holds for the class $\mathbf{IL1}$ with constant $16 \cdot C$.
\end{lemma}

\begin{proof}
We follow exactly the same construction appearing in appendix E of \cite{WWW14}.
Let $ f \in \mathbf{IL1}$. For now we assume $f$ is balanced (i.e. that $\Exp[f] = 0$), and deal with biased functions later. 

Consider the function $g(x, y)$ on $n + k$ variables defined as

\begin{equation}
  g(\mathbf{x},\mathbf{y})=\begin{cases}
    f(\mathbf{x}) & \text{if $\text{AND}(y_1, y_2, ... , y_k) = -1$}\\
    1 & \text{otherwise}
  \end{cases}
\end{equation}

$g$ is extremely biased, as it can get the value $-1$ only when $y_1 = y_2 = ... = y_k = -1 $. By direct calculation, \cite{WWW14} show that:
\[ \Inf[g] = 2^{-k} \cdot (k + \Inf[f]) \]
and also that:

\[ \Ent[g] \geq 2^{-k-3} \cdot (2k + 2 + \Ent[f]) \]

We would like to argue that $g \in \mathbf{AELI}_s$, so we need to pick a large enough $k$ accordingly, so that the following inequality will hold:
\[2^{-k} \cdot (k + \Inf[f]) = \Inf[g] \leq 2^{-s(n+k)} \]

We also want to use self-tensorization on $g$, so we need to ensure $k=o(n)$.
So it would suffice to find $k$ such that:
\begin{itemize}
  \item $k - \log (k + n) \geq s(n+k)$
  \item $k = o(n)$
\end{itemize}

We can pick $k = \max( \sqrt{n}, 2 \cdot s(2n))$. For such $k$ it is clear that $k = o(n)$, and also that $k - \log (k + n) \geq \frac{k}{2} \geq s(2n) \geq s(n+k)$, where the last inequality uses the fact that $k = o(n)$ and that $s$ is monotone increasing - we can assume w.l.o.g that $s$ is a monotone function, or otherwise redefine $\bar{s}$ with $\bar{s}(i) := \max_{i \leq s}{s(i)}$.

So by the fact that $g \in \mathbf{AELI}_s$ and assuming FEI for $\mathbf{AELI}_s$ with universal constant $C$:
\[ 2^{-k-3} \cdot (2k + 2 + \Ent[f]) \leq \Ent[g] \leq C \cdot \Inf[g] = C \cdot 2^{-k} \cdot (k + \Inf[f]) \]
\[ 2k + 2 + \Ent[f] \leq 2^{3} \cdot C \cdot (k + \Inf[f]) \]
\[ \Ent[f] \leq 8 C \cdot \Inf[f] + 8 C k  \leq 8 C \cdot \Inf[f] + o(n) \]

The subclass of balanced functions in $\mathbf{IL1}$ is closed under self-tensorization, so we can use the tensorization technique to get $\Ent[f] \leq 8 C \cdot \Inf[f]$ hereby completing the proof for balanced functions.

If $f \in \mathbf{IL1}$ is biased, we can define
$h(x_1, x_2, ... x_n, x_{n+1}) = x_{n+1} \cdot f(x_1, x_2, ... x_n)$. $h$ is balanced, $\Inf[h] = \Inf[f] + 1$, so $h \in \mathbf{IL1}$, and $\Ent[h] = \Ent[f]$. Therefore we have:
\[ \Ent[f] = \Ent[h] \leq 8C \cdot \Inf[h] \leq 8C \cdot (\Inf[f] + 1) \leq 16 C \cdot \Inf[f] \]
Where the last inequality is the only place where we use the fact that $\Inf[f] \geq 1$ (apart from the fact that $\mathbf{IL1}$ is closed under self-tensorization).
\end{proof}

We would like to extend the lemma from $\mathbf{IL1}$ to $\mathbf{BF}$. If we examine for a moment the class $\mathbf{C} = \left\lbrace f \in \mathbf{BF} : \Ent[f] \geq \sqrt{n} \right\rbrace$, it is easy to see from Weak FEI that for all $f \in \mathbf{C}$, $\Inf[f] \geq \frac{\sqrt{n} }{\log n}$. Therefore $\mathbf{C} \subseteq \mathbf{IL1}$, so FEI holds for $\mathbf{C}$. By Lemma \ref{AELI implies IL1} and Theorem \ref{FEI linear entropy is tight} to be proven in the next section, we can now deduce Theorem \ref{FEI low influence is tight} as a simple corollary:

\begin{theorem} Let $s: \mathbb{N} \to \mathbb{R}$ such that $s(n) = 2^{-o(n)}$. Suppose that FEI holds for some class $\mathbf{AELI}_s$ with universal constant $C$, then FEI holds for the class $\mathbf{BF}$ with constant $16C$.
\end{theorem}

It is natural to ask whether this hardness result extends to FMEI, in the sense that proving FMEI for $\mathbf{AELI}_s$ will imply FMEI for $\mathbf{BF}$. The proof fails because the min-entropy of the original function $f$ vanishes, as $\widehat{f}(\emptyset)$ becomes the largest coefficient of $g$. Furthermore, FMEI is easy for functions with $\Var(f) \leq \frac{1}{2}$, and therefore also functions respecting the stronger condition $\Inf[f] \leq \frac{1}{2}$.

\begin{lemma}
\label{FMEI for biased functions}
Let $ f \colon \{-1,1\}^n \to \{-1,1\}$ such that $\Var[f] \leq \frac{1}{2}$. Then $\Ent_{\infty}[f] \leq 2 \cdot \Inf[f]$.
\end{lemma}

\begin{proof}
\[ \Ent_{\infty}[f] \leq \log \frac{1}{\widehat{f}(\emptyset)^2}
= \log \frac{1}{1 - \Var(f)} \leq \frac{1}{1 - \Var(f)} - 1 \leq 2 \cdot \Var(f) \leq 2 \cdot \Inf[f]. \]
The third inequality makes use of the fact that $\log(x) \leq x - 1$ for $x \geq 1$, and the fourth inequality is due to the fact that $\Var[f] \leq \frac{1}{2}$.
\end{proof}

\section{FEI for Functions With Entropy Linear in n}
In the previous section, we proved FEI for functions with exponentially low influence. We have also matched this with a ``hardness result'', showing that proving FEI for a class of functions with slightly higher influence will imply FEI for all Boolean functions.

These results raise the question of the other non-trivial extremal case - proving FEI for functions with high entropy. As $\Ent[f] \leq n$, a natural interpretation of large entropy could be $\Ent[f] = cn$ for some constant $c \in (0,1)$. In this section, we prove FEI for the class of functions with entropy linear in $n$, and show that improving this to any $o(n)$ will prove FEI for all Boolean functions. We consider the following classes of functions:

\begin{itemize}
	\item The class of functions with linearly-high entropy. For every constant $c > 0$, define $\mathbf{LHE}_c = \bigcup\limits_{n=1}^{\infty} \left\lbrace f \in \mathbf{BF}_n : \Ent[f] \geq cn \right\rbrace$.
	\item The class of functions with ``almost'' linearly-high entropy. For every function $s: \mathbb{N} \to \mathbb{R}$ such that $s(n) = o(n)$, define $\mathbf{ALHE}_s = \bigcup\limits_{n=1}^{\infty} \left\lbrace f \in \mathbf{BF}_n : \Ent[f] \geq s(n) \right\rbrace$.
\end{itemize}

Formally, we show that for any $c \in (0, \frac{1}{2})$ FEI holds for the class $\mathbf{LHE}_c$ with constant $C = O(\frac{1}{c})$. We then show that improving on this result by proving FEI for any $\mathbf{ALHE}_s$ will imply FEI for $\mathbf{BF}$, i.e.\ Conjecture \ref{FEI}.

\subsection{Proving FEI for LHE}

We restate and prove Theorem \ref{FEI high entropy} as follows:
\begin{theorem}
\label{FEI high entropy 2}
Let $c \in (0, \frac{1}{2})$. For all $f \in \mathrm{LHE}_c$, $\Ent[f] \leq \frac{1 + c}{h^{-1}(c^2)} \cdot \Inf[f]$, where $h^{-1}$ is the inverse of the binary entropy function.
\end{theorem}

\begin{proof}
We use the concentration method presented by \cite{CKLS16}, where our partition of the coefficients is of the form $\mathscr{F} = \{S : |S| \leq t \}$. Obviously, we have:
\begin{equation} 
 \label{concentration_for_LHE}
 \Ent[f] = W^{\leq t} \cdot \Ent[\mathscr{F}] + W^{>t} \cdot \Ent[\mathscr{F}^C] + h(W^{\leq t})
\end{equation}

Recall that $\Ent[\mathscr{F}]$ is the entropy of the normalized-to-1 distribution of the squared coefficients of sets in $\mathscr{F}$.
Intuitively, we upper-bound $\Ent[\mathscr{F}]$ by the fact that there are not too many subsets of size $t$ or less. For a constant $\alpha \in (0,1)$ and $t=\alpha n$ we can approximation the volume of the Hamming ball of radius $t$ around $\bar{0}$:
\begin{equation} \label{ball_approx}
B(n, t) = 2 ^ {n \cdot h \left( \frac{t}{n} \right) - o(n)}
\end{equation} 
Therefore, $\Ent[\mathscr{F}] \leq \log B(n, t) \leq n \cdot h \left( \frac{t}{n} \right)$.
For the second and third term of equation \ref{concentration_for_LHE}, we trivially have: $\Ent[\mathscr{F}^C] \leq n$ and $h(W^{\leq t}) \leq 1$.

So combining all of these, we obtain:
\[ \Ent[f] \leq W^{\leq t} \cdot n \cdot h \left( \frac{t}{n} \right) + W^{>t} \cdot n + 1 \]

We start by focusing on functions $f \in \mathrm{LHE}_c$ with $\Ent[f] = cn$, and later extend our proof to all functions in $\mathrm{LHE}_c$, with $\Ent[f] = c'n \geq cn$. Observing that $W^{>t} = 1 - W^{\leq t}$, we obtain:
\begin{equation}
\label{LHE_before_dividing}
cn \leq W^{\leq t} \cdot n \cdot h \left( \frac{t}{n} \right) + (1 - W^{\leq t}) \cdot n + 1
\end{equation} 

We remove the $+1$ term for simplicity, as it is negligible compared to the the other terms (to formalize this, we can replace $cn$ by $0.99 \cdot cn$). Now, dividing equation \ref{LHE_before_dividing} by $n$ and rearranging it:
\[ c \leq W^{\leq t} \cdot h \left( \frac{t}{n} \right) + \left( 1 - W^{\leq t} \right) \]
\[W^{\leq t} \cdot \left( 1 - h \left( \frac{t}{n} \right) \right) \leq (1-c)\]
and finally, 
\[W^{\leq t} \leq \frac{1-c}{1 - h(\frac{t}{n})}\]
Therefore, when picking $h(\frac{t}{n}) = c^2 > 0$ we get 
\[W^{\leq t} \leq \frac{1}{1 + c} \]
Note that by this we picked $t = h^{-1}(c^2) \cdot n$ which is linear in $n$ for our constant $c$, and therefore the Hamming ball volume approximation in equation \ref{ball_approx} is valid. Continuing the computation, it follows that $W^{>t} = 1 - W^{\leq t} \geq 1 - \frac{1}{1 + c}$.

We can lower-bound the influence by $ \Inf[f] \geq  W^{>t} \cdot t$, which is 
\[\Inf[f] \geq (1 - \frac{1}{1 + c}) \cdot h^{-1}(c^2) \cdot n = \frac{c}{1 + c} \cdot h^{-1}(c^2) \cdot n\]
All in all, we get:
\[ \frac{\Ent[f]}{\Inf[f]} \leq \frac{cn}{\frac{c}{1 + c} \cdot h^{-1}(c^2) \cdot n} = 
\frac{1 + c}{h^{-1}(c^2)} \]

So far we proved the inequality for  functions with $\Ent[f] = cn$. For functions with $\Ent[f] = c'n > cn$, we have the same inequality with the bound $\frac{1 + c'}{h^{-1}(c'^2)}$. This function is monotone (decreasing) for the relevant range of $c \in (0, \frac{1}{2})$, and therefore for all functions with $\Ent[f] = c'n \geq cn$ we have $\Ent[f] \leq \frac{1 + c'}{h^{-1}(c'^2)} \cdot \Inf[f] \leq \frac{1 + c}{h^{-1}(c^2)} \cdot \Inf[f]$.
\end{proof}

In \cite{DPV11}, the authors prove FEI for random functions (w.h.p.) with constant $C = 2 + \varepsilon$, arguing the influence of a random function is strongly concentrated around its mean $\frac{n}{2}$. It is known that a random function also has (w.h.p.) entropy linear in $n$. It is even true that w.h.p. $\Ent[f] > (1 - \varepsilon) n$ for any constant $\varepsilon > 0$. This can be shown by applying a Chernoff bound on each of the Fourier coefficients and then a union bound over all $2^n$ coefficients to lower bound the min-entropy. Then $(1 - \varepsilon) n \leq \Ent_{\infty}[f] \leq \Ent[f]$. Therefore, Theorem \ref{FEI high entropy 2} also implies FEI for random functions by a different argument than that of \cite{DPV11}.

\subsection{Proving FEI for ALHE Implies FEI Completely}

\begin{theorem}
Let $s: \mathbb{N} \to \mathbb{R}$ such that $s(n) = o(n)$. Suppose that FEI holds for some class $\mathbf{ALHE}_s$ with universal constant $C$, Then FEI holds for the class $\mathbf{BF}$ with constant $C$.
\end{theorem}

\begin{proof}
Define the $max$ function on two variables
$max(x_1, y_1) = \frac{1}{2} + \frac{1}{2} x_1 + \frac{1}{2} x_2 - \frac{1}{2} x_1 x_2$. It is easy to see that $\Inf[max] = 1$ and $\Ent[max] = 2$.
Also, we recall that the influence, the entropy and the number of variables tensorize nicely: for any $f$ on $n$ variables and $g$ on $k$ variables, the tensor $h(x,y) = f(x) \cdot g(y)$ is a function on $n + k$ variables, with $\Ent[h] = \Ent[f] + \Ent[g]$, $\Inf[h] = \Inf[f] + \Inf[g]$.

Let $ f \colon \{-1,1\}^n \to \{-1,1\}$. 
Consider the function $g(\textbf{x}, \textbf{y}, \textbf{z})$ defined as
\[g(\textbf{x}, \textbf{y}, \textbf{z}) = f(x_1, x_2, ..., x_n) \cdot max(y_1, z_1) \cdot max(y_2, z_2) \cdot ... \cdot max(y_k, z_k) \]
$g$ is a function on $n+2k$ variables, and using tensorization iteratively $k$ times we obtain $\Ent[g] = \Ent[f] + 2k$ and $\Inf[g] = \Inf[f] + k$.

By taking $k = s(2n)$ we have $\Ent[g] \geq s(2n)$ and $g$ is on less than $2n$ variables because $k = o(n)$. So $g \in \mathbf{ALHE}_s$, and therefore by our assumption $\Ent[g] \leq C \cdot \Inf[g]$.
\[ \Ent[f] = \Ent[g] - 2k \leq C \cdot \Inf[g] - 2k = C \cdot (\Inf[f] + k) - 2k = C \cdot \Inf[f] + (C - 2) k \]
If $C \leq 2$ we are done. If $C > 2$, we make use of the fact that $k = o(n)$, and apply the tensorization technique for the class $\textbf{BF}$ to get that $\Ent[f] \leq C \cdot \Inf[f]$.

\end{proof}

We remark that this hardness result extends to FMEI as well. As opposed to the classes discussed in the last section (functions with low influence), in this case the reduction works for FMEI. Multiplying a function by $max(y_i, z_i)$ copies the spectral distribution 4 times, with all (squared) coefficients multiplied by $\frac{1}{4}$. This means that applying the reduction will yield $\Ent_{\infty}[g] = \Ent_{\infty}[f] + 2k$, and by the fact that FMEI tensorizes in the same way as FEI (Lemma \ref{FMEI tenzorizes}) the proof follows.

\begin{theorem}
Let $s: \mathbb{N} \to \mathbb{R}$ such that $s(n) = o(n)$. Suppose that FMEI holds for some class $\mathbf{ALHE}_s$ with universal constant $C$, Then FMEI holds for the class $\mathbf{BF}$ with constant $C$.
\end{theorem}

\section{FEI for Functions With Constant L1 Fourier Norm}
As mentioned previously, it is shown in \cite{CKLS16} that the Fourier entropy of a function is bounded by the logarithm of $\|\widehat{f}\|_1$. Specifically they show that $\Ent[f] \leq 4 \log \|\widehat{f}\|_1 + 9 $. We mention two other ways in which this can be seen that provide a better constant:
\begin{itemize}
  \item Lemma 6.9 in \cite{GSTW16} taken with $p=1$, states that for any non-negative $(p_1,p_2,...,p_m)$ that sum up to $1$, we get $\sum_{i=1}^{m} p_i \log \frac{1}{p_i} \leq 2 \log \sum_{i=1}^{m} \sqrt{p_i}$. Plugging in the distribution $\widehat{f}(S)^2$, we get exactly $\Ent[f] \leq 2 \log \|\widehat{f}\|_1$.
  
  \item  Recalling the definition of the Renyi entropy over the distribution $\mathcal{X}$ of the squared Fourier coefficients $\widehat{f}(S)^2$, we get $\Ent_1[\mathcal{X}] = \Ent[f] = \sum_{S \subseteq [n]} \widehat{f}(S)^2 \log \frac{1}{\widehat{f}(S)^2}$, and $\Ent_{\frac{1}{2}}[\mathcal{X}] = 2 \log \sum_{S \subseteq [n]} |\widehat{f}(S)| = 2 \log \|\widehat{f}\|_1$. By the fact that for $a > b$ we get  $\Ent_{a}[\mathcal{X}] < \Ent_{b}[\mathcal{X}]$, we obtain $\Ent[f] \leq 2 \log \|\widehat{f}\|_1$.
\end{itemize}

The main caveat of these results is that they do not show FEI for the class of functions with constant $\|\widehat{f}\|_1$, denoted
$\textbf{CL1}_L = \bigcup\limits_{n=1}^{\infty} \left\lbrace f \in \mathbf{BF}_n : \|\widehat{f}\|_1 \leq L \right\rbrace$, because the influence of functions in this class can be arbitrarily small. By making subtle changes to the proof given by \cite{CKLS16} we overcome this and prove FEI for $\textbf{CL1}_L$ with universal constant $C = 4 \log L + 21$.

\begin{theorem}
Let $ f \colon \{-1,1\}^n \to \{-1,1\}$ be a Boolean function with $\|\widehat{f}\|_1 = L$. Then
${\Ent[f] \leq (4 \log L + 11) \Var(f)} + 10 \cdot \Inf[f] $, and in particular 
$\Ent[f] \leq (4 \log L + 21) \cdot \Inf[f] $.
\end{theorem}

\begin{proof}
Let $\theta := (\frac{\Var(f)}{4L})^2 $. We divide the Fourier coefficients into three sets:
\begin{enumerate}
\item $Z = \{ \emptyset \} $
\item $G_{\text{large}} = \{ S : |\widehat{f}(S)| > \theta, S \neq \emptyset \} $
\item $G_{\text{small}} = \{ S : |\widehat{f}(S)| \leq \theta, S \neq \emptyset \} $
\end{enumerate}

Now, we look separately at the terms of the entropy according to this partition:
\begin{equation}
\label{entropy_in_three_parts}
\Ent[f] = 
\widehat{f}(\emptyset)^2 \log \frac{1}{\widehat{f}(\emptyset)^2} + 
\sum_{S \subseteq G_{\text{large}}} \widehat{f}(S)^2 \log \frac{1}{\widehat{f}(S)^2} +
\sum_{S \subseteq G_{\text{small}}} \widehat{f}(S)^2 \log \frac{1}{\widehat{f}(S)^2}
\end{equation}

To bound the first term of equation \ref{entropy_in_three_parts}, we note that $\widehat{f}(\emptyset)^2 + \Var(f) = 1$.
\begin{enumerate}
\item If $\widehat{f}(\emptyset)^2 \leq \frac{1}{2} \leq \Var(f)$ we can bound $\widehat{f}(\emptyset)^2 \log \frac{1}{\widehat{f}(\emptyset)^2} \leq 1 \leq 2 \Var(f)$.
\item If $\widehat{f}(\emptyset)^2 \geq \frac{1}{2} \geq \Var(f)$, then $\widehat{f}(\emptyset)^2 \log \frac{1}{\widehat{f}(\emptyset)^2} \leq \Var(f) \log \frac{1}{\Var(f)}$, recalling that for $p \geq \frac{1}{2}$, we have $p \log \frac{1}{p} \leq (1-p) \log \frac{1}{1-p}$.
\end{enumerate}
Accounting for the two possibilities, we get
$\widehat{f}(\emptyset)^2 \log \frac{1}{\widehat{f}(\emptyset)^2} \leq (2 + \log \frac{1}{\Var(f)}) \cdot \Var(f) $.

To bound the second term of equation \ref{entropy_in_three_parts}, we note that for $S \in G_{\text{large}}$, $\log \frac{1}{\widehat{f}(S)^2} \leq \log \frac{1}{\theta^2}$.
\[ \sum_{S \subseteq G_{\text{large}}} \widehat{f}(S)^2 \log \frac{1}{\widehat{f}(S)^2} \leq
\log \frac{1}{\theta^2} \cdot \sum_{S \subseteq G_{\text{large}}}{ \widehat{f}(S)^2 } \leq
(4\log L + 8 + 4 \log \frac{1}{\Var(f)}) \cdot \Var(f) \]

To bound the third term of equation \ref{entropy_in_three_parts}, note that for a Boolean function it holds that $\Var(f) \leq 1$ and $\|\widehat{f}\|_1 \geq 1$, hence $\theta < \frac{1}{16}$. Also note that for $x > 16$, $\log x < \sqrt{x}$, and therefore $\log \frac{1}{\widehat{f}(S)} < \frac{1}{\sqrt{\widehat{f}(S)}}$.
\begin{align*}
\sum_{S \in G_{small}} \widehat{f}(S)^2 \log \frac{1}{\widehat{f}(S)^2}
& \leq 2 \sum_{S \in G_{\text{small}}} \widehat{f}(S)^2 \frac{1}{\sqrt{\widehat{f}(S)}} \\
& \leq 2 \max_{S \in G_{\text{small}}}\sqrt{|\widehat{f}(S)|} \cdot \sum_{S \in G_{\text{small}}} |\widehat{f}(S)| \\
& \leq 2 \sqrt{\theta} \cdot L = \frac{\Var(f)}{4L} \cdot 2L = \frac{1}{2} \Var(f)
\end{align*}

Plugging these three bounds into equation \ref{entropy_in_three_parts}:
\[ \Ent[f] \leq
(2 + \log \frac{1}{\Var(f)}) \cdot \Var(f) + (4\log L + 8 + 4 \log \frac{1}{\Var(f)}) \cdot \Var(f) + \frac{1}{2} \Var(f)
\]
rearranging the inequality and applying Lemma \ref{Edge Isoperimetric Inequality 2} we get:
\[ \Ent[f] \leq (4 \log L + 11 + 5 \log \frac{1}{\Var(f)}) \cdot \Var(f)
\leq (4 \log L + 11) \Var(f) + 10 \Inf[f] \]
in particular, we obtain:
 \[\Ent[f] \leq (4 \log L + 21) \Inf[f] .\]
\end{proof}

\begin{fact}
For a Boolean function $ f \colon \{-1,1\}^n \to \{-1,1\}$, let $L_c(f)$ be the size of a minimal sub-cube partition of $f$,
$deg(f)$ be the degree of $f$ as a real valued polynomial, $granularity(f)$ be the granularity of $f$, $sparsity(\widehat{f})$ be the size of $supp(\widehat{f})$, $D(f)$ be the minimal depth of a decision tree computing $f$, and
$l(f)$ be the minimal number of leaves of a decision tree computing $f$. Then:

\begin{itemize}
  \item $\|\widehat{f}\|_1 \leq L_c(f)$, (Lemma 4.8(i) in \cite{CKLS16})
  \item $\|\widehat{f}\|_1 \leq l(f) \leq 2^{D(f)}$ (Proposition 3.16 in \cite{O'DBook}). This is still true when allowing parity queries at each node.
  \item $\|\widehat{f}\|_1 \leq 2^{deg(f)}$  (\cite{CKLS16})
  \item $\|\widehat{f}\|_1 \leq 2^{granularity(f)}$ (from Parseval's identity)
  \item $\|\widehat{f}\|_1 \leq \sqrt{sparsity(\widehat{f})}$  (from Parseval's identity)
  
\end{itemize}
\end{fact}

This means that functions with constant degree, for instance, are a subclass of the functions with constant L1 spectral norm, and FEI holds for them as well with the appropriate constant.

\begin{corollary}
FEI holds for functions with constant sub-cube partition, constant degree, constant decision tree depth, constant decision tree size, constant granularity or constant sparsity.
\end{corollary}

\section{Protocol Based Approach to FEI}
In \cite{WWW14}, the authors suggest an insightful perspective on FEI: they note that $\Ent[f] \leq C \cdot \Inf[f]$ is true for a given function if there exists a communication protocol, that given a random subset $S \subseteq [n]$ sampled according to the spectral distribution of $f$, can communicate the value of S using at most $C \cdot \Inf[f]$ bits in expectation. They offer such protocols for functions with decision trees of constant average depth (also requiring $\Inf[f] \geq 1$), and for functions with read-k decision trees ($k$ is constant) - proving FEI for both classes. The result regarding average decision tree depth is also shown using more elementary methods in \cite{CKLS16} (and with a better constant), but the result for read-k was not previously known - until their work, only read-once decision trees and read-once formulas were conquered, mainly by inductions that relied heavily on the fact that the different parts of the decision tree or formula are independent of each other. 

It is unclear how this method could be harnessed to prove FEI for less structured classes of functions than decision trees. For example, If we consider symmetric functions there is no obvious structure to the Fourier coefficients we can exploit - even if we know the size of the set we need to send, $|S|$, all coefficients within that level are of equal size. Therefore it is reasonable to assume a protocol will somehow encode the size of $S$ (denote $|S| = t$) and the index of the $S$ in $\left \lbrace S : |S| = t \right \rbrace$. To prove that this protocol is indeed efficient, we must have some deeper understanding of the distribution of weight between Fourier levels of symmetric functions (as done in \cite{OWZ11}), so the protocol perspective does not seem to help in this case.

That being said, we are hopeful this method could be applied to more classes of functions with useful structure. Promising candidates could be circuits and formulas, and specifically DNFs. This will have the additional application of proving a version of Mansour's conjecture, as explained in \cite{OWZ11}. A first natural step would be to examine read-k DNFs - we note that Mansour's conjecture for read-k DNFs was proven by Klivans et al. \cite{KLW10}.
We suggest a protocol that may prove FEI for read-k DNFs, and independently we provide a proof of FMEI for ``regular'' read-k DNFs. By regular we mean that all clauses are of the same size (up to a constant multiplicative factor), and the number of clauses is exponential in the clause size.

\subsection{Towards FEI for Read-k DNFs}
As an example to the intuitive power of the protocol method, we examine a natural protocol that works for the Tribes function $\Tr$, which is a read-once DNF. We denote by $w$ the number of variables in a tribe, $s = \Theta(2^w)$ the number of tribes, and $n = sw$ the total number of variables. For the following protocol, we denote by $\text{protocol}(S)$ the number of bits used in the protocol for a set $S$.

\begin{center}
\fbox{\begin{minipage}{30em}

  \begin{center}
  \textbf{Given} $S \subseteq [n]$: 
  \end{center}
  
  \begin{enumerate}
	\item If $S = \emptyset$, output nothing.
	\item For each tribe with non-empty intersection with $S$, output the tribe index ($i \in [s]$) and 0-1 string of length $w$, denoting for each variable in $T_i$ whether it is in $S$ or not.
	\item Terminate with a $\bot$.
  
  \end{enumerate}
  
\end{minipage}}

\end{center}

We recall (see Section 4.2 in \cite{O'DBook}) that for $S \neq \emptyset$ that has non-empty intersection with $l$ tribes,
$\widehat{\Tr}(S)^2 = 4 \cdot 2^{-2lw} \cdot (1-2^{-w})^{2(s-l)}$. Furthermore, it is easy to see that for $S$ intersecting with $l$ tribes, $\text{protocol}(S) = l \cdot (\log (s) + w)$.

The following calculation shows that the expected protocol length is indeed $O(\Inf[\Tr])$, by summing over the sets $S$ according to the number of tribes they intersect:

\begin{align*}
\Exp[\text{protocol length}] =
& \sum_{S \subseteq [n]} \widehat{\Tr}(S)^2 \cdot \text{protocol}(S) \\
& = \sum_{l=1}^{s}  \left( \binom{s}{l} (2^w - 1)^l ) \cdot (4 \cdot 2^{-2lw} \cdot (1-2^{-w})^{2(s-l)}) \cdot ( l \cdot (\log (s) + w) \right) \\
& \leq 4 \cdot \sum_{l=1}^{s}  \binom{s}{l} \cdot 2^{-lw} \cdot (1-2^{-w})^{2(s-l)} \cdot l \cdot O(w) \\
& \leq O(w) \cdot \sum_{l=1}^{s}  \binom{s}{l} \cdot (2^{-w})^l \cdot (1-2^{-w})^{s-l} \cdot l \\
& = O(w) \cdot \Exp[ \text{Bin}(s, 2^{-w})] = O(w)
\end{align*}

It is well known that $\Inf[\Tr] = \Theta(\log(n)) = \Theta(w)$ (see Proposition 4.13 in \cite{O'DBook}), so indeed $\Exp[\text{protocol length}] \leq O(\Inf[\Tr])$.

We note (without proof) that this protocol can be extended to general read-once DNFs, which FEI is already known for, by assigning variable-length encodings to each of the tribes, based on their size. In fact, this is the protocol that we get when using the composition of protocols for OR and AND, as explained in Section 3 of \cite{WWW14}. For simplicity, from now on we assume regular DNFs, defined as follows:

\begin{definition*} Let $C_1,C_2 > 0$. We say $f$ is a $(C_1, C_2)$-regular DNF (or just ``regular''), if there exists some $w \in \mathbb{N}$ s.t. $f = T_1 \lor T_2 \lor ... \lor T_s$ s.t. the number of variables in each clause respects $C_1 w \leq \text{size}(T_i) \leq w$, and the number of clauses is $s = 2^{C_2 n}$.
\end{definition*}
For example, the (essentially unbiased) Tribes function $\Tr$ is a regular read-once DNF with $C_1 = 1$, $C_2 \thickapprox 1$.

\subsubsection{A Suggested Protocol for Read-k DNFs}
Recall the strategy of \cite{WWW14} for decision trees: for read-once decision tree, they note that every $S$ with non-zero weight has exactly one path from the root that contains its variables (and maybe additional variables), and encode the path efficiently. For read-k decision tree, they note that every $S$ with non-zero weight has at least one path, but potentially many - some may be short, and some much longer - and intuitively, that if $S$ has large Fourier weight then it has some short path  containing it (and the weight of $S$ ``comes'' from these short paths).

Analogously, we can view our protocol for Tribes as sending a ``set cover'' of tribes, in the sense that $S$ is a subset of their union. For a read-once DNF, $S$ has exactly one cover, and from the Fourier coefficients of $\Tr$, it is obvious that sets $S$ with large Fourier coefficients have small covers that can be sent efficiently.
For read-k DNFs, $S$ may have many covers. Informally, from looking at the structure of a DNF as a polynomial, we can see that the Fourier weight of $S$ gets contribution from the different covers of $S$, with most weight contributed by the small covers. As we are dealing with read-k DNFs, the number of covers of $S$ can be bounded, and we could hope the $\widehat{f}(S)^2$ is dominated by its smallest cover. We therefore conjecture the following protocol for read-k DNFs: 

\begin{center}
\fbox{\begin{minipage}{30em}

  \begin{center}
  \textbf{Given} $S \subseteq [n]$: 
  \end{center}
  
  \begin{enumerate}
	\item If $S = \emptyset$, output nothing.
	\item Let $\mathcal{C} \subseteq \left\lbrace T_i \right\rbrace _i \in [s]$ be the \textit{smallest} cover of $S$. For each tribe $T_i \in \mathcal{C}$ output the tribe index and a 0-1 string of length $w$, denoting for each variable in $T_i$ whether it is in $S$ or not.
	\item Terminate with a $\bot$.
  
  \end{enumerate}
  
\end{minipage}}

\end{center}

To specify this protocol completely, we still need to define the size of the cover in order for the term ``smallest'' cover to be meaningful. Natural options can be the number of tribes involved - $|\mathcal{C}|$, the combined sizes of the tribes $\sum_{T_i \in \mathcal{C}} |T_i|$, or the number of unique variables involved, $|\bigcup_{T_i \in \mathcal{C}} T_i |$. We conjecture the latter to be the correct measure (from inspecting several DNFs and Fourier weights of some sets $S$).

\begin{conjecture}
The expected price of the above protocol for a read-k DNF $f$ is $O_k(\Inf[f])$.
\end{conjecture}

\subsubsection{FMEI for Regular Read-k DNFs}

As a modest first step, we prove FMEI for the class of regular read-k DNFs. When considering min-entropy, only the weight of the largest Fourier coefficient matters, and as can be derived from the discussion, we expect it to be some $S$ that can be covered by using only one tribe. 
In Blum et al. \cite{BFJKMR94}, while discussing the learnability of read-k DNFs, the authors implicitly show the following lemma that formalizes the latter notion regarding the weight of sets covered by a single tribe:

\begin{lemma} \label{weight of DNF inside clause}
For any read-k DNF $f$ there is a family $\mathcal{F} \subseteq 2^{[n]}$ with $|\mathcal{F}| \leq 24 n^2 k^2$ such that 
$\sum_{S \in \mathcal{F}} \widehat{f}(S)^2 \geq \frac{1}{4k}$.
\end{lemma}

The family of sets the lemma refers to is $\mathcal{F} = \bigcup_{|V_i| \leq \log(24kn)} P(V_i)$ where $V_i$ is the set of variable in clause $T_i$ and $P(X)$ is the power set of $X$. In words, these are sets that can be covered by at most one tribe that is not too large (very wide terms barely affect the function or it's Fourier structure).

\begin{lemma} \label{entropy of regular DNF is small}
Let $f$ be a regular read-k DNF with width $w$. Then $\Ent_{\infty}[f] \leq O(w + \log k)$.
\end{lemma}

\begin{proof}
By Lemma \label{weight of DNF inside clause}, we know there is a set $\mathcal{F}$ of size less than $24 n^2 k^2$ s.t. $\sum_{S \in \mathcal{F}} \widehat{f}(S)^2 \geq \frac{1}{4k}$. Hence,

\[\frac{1}{4k} \leq \sum_{S \in \mathcal{F}} \widehat{f}(S)^2 \leq
24 n^2 k^2 \max_{S \in \mathcal{F}} \widehat{f}(S)^2 \]
and it follows that
\[\max_{S \in \mathcal{F}} \widehat{f}(S)^2 \geq \frac{1}{96k^3n^2} \]

\[\Ent_{\infty}[f] = \min_{S \subseteq [n]}\ \log \frac{1}{\widehat{f}^2(S)}\ \leq
\log \frac{1}{\max_{S \in \mathcal{F}} \widehat{f}(S)^2} \leq 7 + 3 \log k + 2 \log n = O(w + \log k) \]

In the last step (and only there) we use the fact that $f$ is a regular DNF, and therefore  $\log n \leq \log (s \cdot w) = C_2 w + \log w \leq O(w)$.
\end{proof}

\begin{lemma} \label{influence of regular DNF is large}
Let $f$ be a regular read-k DNF with width $w$. Then $\Inf[f] \geq \Omega(w) - \log k$.
\end{lemma}

\begin{proof}
We show that for a regular read-k DNF the influence of any single variable must be small, and then by using KKL we lower bound the total influence. For any variable $x_j$, to be influential for some input, it must change the value of a clause it appears in. To influence clause $T_i$, all other variables of $T_i$ must be assigned the value True, so $x_j$ is influential through that clause for at most a $2^{-|T_i| + 1} \leq 2^{-C_1 w + 1}$ fraction of all inputs ($C_1 w$ is the lower bound for clause size). $x_j$ appears in at most $k$ clauses, so $\Inf_j[f] \leq k \cdot 2^{-C_2 w + 1}$. From the Edge Isoperimetric version of KKL (see Section 10.3 of \cite{O'DBook}), $\Inf[f] \geq \log \frac{1}{\mathbf{MaxInf}[f]} \geq C_2 w - 1 - \log k = \Omega(w) - \log k$.

\end{proof}

Combining Lemmas \ref{entropy of regular DNF is small} and \ref{influence of regular DNF is large}, and viewing $k$ as constant, we obtain FMEI for regular read-k DNFs.

\begin{corollary}
Let $f$ be a regular read-k DNF, then $\Ent_{\infty}[f] = O(\Inf[f])$.
\end{corollary}

It is very probable that with some finer arguments, this can be generalized to read-k DNF for constant $k$.

\subsection{On the Covariance of Read-k Decision Trees}

In \cite{WWW14}, the authors define the following covariance measure on a decision tree $T$ computing $f$:
\begin{equation}
\label{cov by induction}
\Cov[T] = \Cov(g,h) + \frac{1}{2} ( \Cov[T_0] + \Cov[T_1] )
\end{equation} 
where $T_0$ and $T_1$ are the left and right sub-trees of $T$, with the corresponding functions $g,h$ on the variables $x_1, x_2, ... x_{n-1}$, assuming w.l.o.g that at the root the variable that is queried is $x_n$.
They show that for a function with a read-k decision tree, it holds that $\Cov[T] \leq (k-1) \cdot \Var(f)$ which then implies $\Cov[T] \leq (k-1) \cdot \Inf[f]$. Using this fact with a protocol they described, they prove FEI for read-k decision trees with $C=\Theta(k)$.
Although for the matter of proving FEI we view $k$ is a constant, it is interesting to find the real bound relating the $\Cov[T]$ and $\Inf[f]$ for a read-k decision tree, and it is not clear that $k-1$ is the correct coefficient. The authors present an example where the coefficient is only $\log k$, and explicitly conjecture that this is tight - that for any $f$ that is computable by a read-k decision tree, $\Cov[T] \leq \log k \cdot \Var[f]$. We provide a step in this direction showing that $\Cov[T] \leq \sqrt{k} \cdot \Inf[f]$. This will in turn improve the FEI coefficient for read-k decision tree to $C=\Theta(\sqrt{k})$.

\begin{theorem}
Let $ f $ be computed by a read-k decision tree. Then $\Cov[T] \leq 2 \cdot \sqrt{k} \cdot \Inf[f]$.
\end{theorem}

Our proof will follow the lines of the original proof of \cite{WWW14}, which is a structural induction. To show that $\Cov[T] \leq (k-1) \cdot \Var(f)$, something stronger is actually shown:
\[\Cov[T] \leq \sum_{\emptyset \neq S \subseteq [n]}{(m_T(S) - 1) \cdot \widehat{f}(S)^2} \]
where $m_T(S) = \max_{i \in S} (a_i(T))$, and $a_i(T)$ is the number of appearances of $i$ in $T$. It is obvious that $m_T(S) \leq k$, and their theorem follows.

To show the $\Cov[T] \leq \log k \cdot \Inf[f]$, it would make sense to refine $m_T(S)$ into a more suitable measure. For example, to define:
\[l_T(S) = \sum_{i \in S} \log(a_i(T)).\]
A technical note - this is not defined if any variable in $S$ has $a_i(T) = 0$, but that means it doesn't appear in $T$ and therefore $\widehat{f}(S)^2=0$, so this will not be a problem.
As $l_T(S) \leq |S| \cdot \log k$, it would be enough to show that 
\[\Cov[T] \leq \sum_{\emptyset \neq S \subseteq [n]}{|S| \log k \cdot \widehat{f}(S)^2}  = \log k \cdot \Inf[f]. \]
Sadly, the proof did not follow through with $l_T(S)$, so we had to compromise and use the following definition:
\[ sq_T(S) = 2 \cdot \sum_{i \in S} { \sqrt{a_i(T)}} \leq 2 \cdot \sqrt{k} \cdot |S| . \]
The 2 factor is needed for technical reasons. Our proof is by structural induction on $T$ that
\begin{equation}
\label{induction target}
\Cov[T] \leq \sum_{\emptyset \neq S \subseteq [n]}{sq_T(S) \cdot \widehat{f}(S)^2}
\end{equation}
and then it will follow that $\Cov[T] \leq 2 \sqrt{k} \cdot I[f]$.

\begin{proof}

\textbf{Base case:} $f$ is a tree with one variable, and the left and right sub-trees are constant. Therefore the left-hand side of Inequality \ref{induction target} is $0$ and the right-hand side is always non-negative.

\textbf{Inductive step:} Suppose w.l.o.g the root variable of $T$ is $x_n$. We use the recursive definition of $\Cov[T]$ given by Equation \ref{cov by induction}:
\[ 2 \cdot \Cov[T] = 2 \cdot \Cov(g,h) + (\Cov[T_0] + \Cov[T_1]) \]

We focus on the first term. Let $J$ be the set of coordinates which appear in both sub-trees $T_0, T_1$. Because $x_n$ is the root variable, it doesn't appear in either $T_0, T_1$, so $J$ is a subset of $[n - 1]$.

\[ 2 \cdot \Cov(g,h) = 2 \sum_{\emptyset \neq S \subseteq [n]}{\widehat{g}(S) \widehat{h}(S)} \
= 2 \sum_{\emptyset \neq S \subseteq J}{\widehat{g}(S) \widehat{h}(S)}. \]
This is exactly the same bound as in \cite{WWW14}, but we stop before their last step (they bounded $2 \widehat{g}(S) \widehat{h}(S)$ by $\widehat{g}(S)^2 + \widehat{h}(S)^2$), which is potentially wasteful.

To bound the second term, we apply the inductive hypothesis:
\[ \Cov[T_0] + \Cov[T_1] \leq \sum_{\emptyset \neq S \subseteq [n-1]}{sq_{T_0}(S) \cdot \widehat{g}(S)^2 + sq_{T_1}(S) \cdot \widehat{h}(S)^2}. \]

For $S \nsubseteq J$, we do not get anything added from the first term, so it is enough to notice that $sq_{T_i}(S) \leq sq_{T}(S)$ and we get 
\[sq_{T_0}(S) \cdot \widehat{g}(S)^2 + sq_{T_1}(S) \cdot \widehat{h}(S)^2 \leq sq_{T}(S) \cdot (\widehat{g}(S)^2 + \widehat{h}(S)^2).\]

For $S \subseteq J$, we need to add the $2 \widehat{g}(S) \widehat{h}(S)$ term.
So we would like to show that 
\begin{equation}
\label{hard inequality}
2 \widehat{g}(S) \widehat{h}(S) + sq_{T_0}(S) \cdot \widehat{g}(S)^2 + sq_{T_1}(S) \cdot \widehat{h}(S)^2 \leq sq_{T}(S) \cdot (\widehat{g}(S)^2 + \widehat{h}(S)^2).
\end{equation}

We pause our proof for a short intuitive discussion regarding the last inequality. If we could show that $sq_{T_i}(S) \leq sq_{T}(S) - 1$ for both $T_0, T_1$ that would be enough, by upper bounding $2\widehat{g}(S) \widehat{h}(S) \leq \widehat{g}(S)^2 + \widehat{h}(S)^2$ as done by \cite{WWW14}. Note that if the appearances of variables in $S$ are divided ``more or less equally'' between $T_0$ and $T_1$, then this can be done easily - in this case, the argument even carries over even for $l_T(S)$. For instance, if one variable from $S$ is split equally amongst the two sub-trees, we have $l_{T_i}(S) \leq l_{T}(S) - 1$. Another example, is when two variables from $S$ ``disagree'' on the sub-tree where they appear more often, and also in this case $l_{T_i}(S) \leq l_{T}(S) - 1$ for both $T_0, T_1$. 

The challenging case is when we have all variables in $S$ appear almost exclusively in one sub-tree. The most extreme instance of this case is when we have $k$ appearances of every variable from $S$ in $T$, only one appearance for each variable in $T_0$, and $k-1$ appearances in $T_1$. Note that even in this extreme case we cannot deduce anything about the ratio of $\widehat{g}(S)^2$ and $\widehat{h}(S)^2$ - even though $h$ has many variables of $S$, they may hide deep inside the tree while in $g$ they can form a path close to the root, and contribute significantly to the Fourier weight of $S$. This is where switching to $sq_{T}(S)$ helps, and we will prove this soon.

Assuming for a moment Inequality \ref{hard inequality}, we finish the proof exactly as done in \cite{WWW14}:

\begin{align*}
2 \cdot \Cov(g,h) + \Cov[T_0] + \Cov[T_1] \leq
& \sum_{\emptyset \neq S \subseteq [n-1]}{sq_{T}(S) \cdot ( \widehat{g}(S)^2 + \widehat{h}(S)^2)} \\
& \leq 2 \cdot \sum_{\emptyset \neq S \subseteq [n-1]}{sq_{T}(S) \cdot (\widehat{f}(S)^2 + \widehat{f}(S \cup \lbrace n \rbrace)^2)} \\
& \leq 2 \cdot \sum_{\emptyset \neq S \subseteq [n]}{sq_{T}(S) \cdot \widehat{f}(S)^2}
\end{align*}
where the last step is due to the fact $sq_{T}(S) \leq sq_{T}(S \cup \lbrace n \rbrace)$, and the one before is due to Proposition 2.3 of \cite{WWW14} stating that for a decision tree of a function $f$ with root $x_n$ and sub-trees computing $g,h$ the following holds:  $\widehat{g}(S)^2 + \widehat{h}(S)^2 = 2 (\widehat{f}(S)^2 + \widehat{f}(S \cup \lbrace n \rbrace)^2)$.

So all that is left to prove is Inequality \ref{hard inequality} for $S \subseteq J$. By simple calculation, $\sqrt{l} - \sqrt{l-1} > \frac{1}{2\sqrt{l}}$, and in general $\sqrt{l} - \sqrt{l-c} > \frac{c}{2\sqrt{l}}$, and for $l \leq k$, then obviously $\sqrt{l} - \sqrt{l-c} > \frac{c}{2\sqrt{k}}$.

If $g$ has a variable $i$ with
$a_i(T_0) \leq a_i(T) - 2\sqrt{a_i(T)}$, then $\sqrt{a_i(T)} - \sqrt{a_i(T_0)} \geq 1$, and hence by definition of $sq_{T}$, we get $sq_{T_0}(S) \leq sq_{T}(S) - 1$. If $h$ has such variable too, we can use the bound $2\widehat{g}(S) \widehat{h}(S) \leq \widehat{g}(S)^2 + \widehat{h}(S)^2$ and we are done.

We are left with the case where (w.l.o.g) all variables in $S$ ``tend to $T_1$'', but also appear at least once in $T_0$. Formally, for any $i \in S$:
\begin{itemize}
\item $a_i(T_0) \leq 2\sqrt{a_i(T)}$
\item $a_i(T_1) \leq a_i(T) - 1$
\end{itemize}

We write Inequality \ref{hard inequality} (which we need to prove) a bit differently:
\[ 2 \widehat{g}(S) \widehat{h}(S) 
\leq (sq_{T}(S) - sq_{T_0}(S)) \cdot \widehat{g}(S)^2 + (sq_{T}(S) - sq_{T_0}(S)) \cdot \widehat{h}(S)^2. \]

\begin{itemize}
\item $sq_{T}(S) - sq_{T_0}(S) \geq 2 \cdot \sum_{i \in S}{ \left( \sqrt{a_i(T)} - \sqrt{2}\sqrt[4]{a_i(T)} \right)  } \geq 
2 \cdot \frac{1}{2} \cdot \sum_{i \in S}{ \sqrt{a_i(T)} } = \sum_{i \in S}{ \sqrt{a_i(T)} }$.
The first inequality is because in this case, $a_i(T_0) \leq 2\sqrt{a_i(T)}$. The second inequality is correct if all $a_i(T)$ are larger than 16 - this is merely a technical detail, that can be fixed for smaller $a_i(T)$, for example, by defining $a_i(T)$ as the number of appearances of $i$ in $T$ plus $16$, which does not affect our bound asymptotically - we ignore this for the simplicity of this proof.
\item $sq_{T}(S) - sq_{T_1}(S) \geq 2 \cdot \sum_{i \in S}{ \sqrt{a_i(T)} - \sqrt{a_i(T) - 1} } \geq 
2 \cdot \frac{1}{2} \cdot \sum_{i \in S}{ \frac{1}{\sqrt{a_i(T)}}} = \sum_{i \in S}{ \frac{1}{\sqrt{a_i(T)}}}$.
The first inequality comes from the fact $a_i(T) - 1 \geq a_i(T_1)$.
\end{itemize}

Now we split the weight of $2\widehat{g}(S) \widehat{h}(S)$, but not necessarily to $\widehat{g}(S)^2 + \widehat{h}(S)^2$. For any non-zero $m \in \mathbb{R}$, we have $2\widehat{g}(S) \widehat{h}(S) \leq m^2 \widehat{g}(S)^2 + \frac{1}{m^2} \widehat{h}(S)^2$. This is due to $0 \leq (m \widehat{g}(S) - \frac{1}{m} \widehat{f}(S)) ^2$. We want to pick $m^2$ such that 
\[ 2\widehat{g}(S) \widehat{h}(S) \leq m^2 \widehat{g}(S)^2 + \frac{1}{m^2} \widehat{h}(S)^2 \leq
\sum_{i \in S}{ \sqrt{a_i(T)} } \cdot \widehat{g}(S)^2 + \sum_{i \in S}{ \frac{1}{\sqrt{a_i(T)}}} \cdot \widehat{h}(S)^2 .\]
As we cannot bound the ratio of $\widehat{g}(S)^2, \widehat{h}(S)^2$, we satisfy the two separate inequalities:
\[ \begin{cases} m^2 \leq \sum_{i \in S}{ \sqrt{a_i(T)}} \\
\frac{1}{m^2} \leq \sum_{i \in S}{ \frac{1}{\sqrt{a_i(T)}}}  \end{cases} \]
Which is equivalent to:
\[\frac{1} {\sum_{i \in S}{ \frac{1}{\sqrt{a_i(T)}}}} \leq m^2 \leq \sum_{i \in S}{ \sqrt{a_i(T)}} \]
Picking $m^2 = \sum_{i \in S}{ \sqrt{a_i(T)}} $, the last inequality is held due to the ``arithmetic mean is larger than harmonic mean'' theorem.

This covers all the cases of the induction step, therefore we have 
\begin{equation}
\label{induction target}
\Cov[T] \leq \sum_{\emptyset \neq S \subseteq [n]}{sq_T(S) \cdot \widehat{f}(S)^2}
\end{equation}
and then $\Cov[T] \leq 2 \sqrt{k} \cdot I[f]$, concluding the proof.

\end{proof}

\section{Improved Bound on the L1 Norm of Decision Trees}
It is known that for any Boolean function $f$ that is computed by a decision tree $T$,  $ \|\widehat{f}\|_1 \leq \text{size}(T)$ (see Proposition 3.16 in \cite{O'DBook}), where $\text{size}(T)$ is the number of leaves of the tree. We provide the following stronger bound. 
\begin{proposition}
For a Boolean function $f$ that is computed by a decision tree $T$:
\[ \|\widehat{f}\|_1 \leq \text{boundary\_size}(T) - \sum_{v \in \text{inner}(T)}{| \Cov(g_v, h_v) | } \]
where $\text{boundary\_size}(T)$ is the number of nodes that have at least one child that is a leaf. The sum of $\Cov(g_v, h_v)$ is over all inner nodes of T, i.e. nodes that have two non-leaf children.
\end{proposition}

This is stronger than the original bound in two senses: the first, $\text{boundary\_size}(T) \leq \text{size}(T)$, and can be as small as $\frac{\text{size}(T)}{2}$ for a ``full'' binary tree. The second is that we subtract a non-negative term that can be significant. Examine the standard and the new bounds for the parity function on $n$ variables, $f = \chi_{[n]}$. Obviously, $\|\widehat{f}\|_1 = 1$. It is clear that for the natural decision tree computing $f$, $\text{size}(T) = 2^{n+1}$ is a terrible bound. On the other hand $\text{boundary\_size}(T) = 2^{n}$, the number of inner nodes is $2^{n} - 1$ and at any inner node $| \Cov(g_v, h_v) | = 1$. Therefore, the old inequality gives a bound of $2^{n+1}$ while the new one gives an exact bound of $2^{n} - (2^{n} - 1) = 1$. As another example, the new bound is also exact for the Address function - all covariances are 0, but the boundary size is a tight bound. This is also a much better bound for $\text{OR}_n$ and $\text{AND}_n$. 

\begin{proof}
Let us examine a function $f$ computed by a tree with root $x_n$, with left function $g$ and right function $h$. We can write
$f(x) = \frac{1+x_n}{2} \cdot g(x) + \frac{1-x_n}{2} \cdot h(x)$. It is easy to see that for any $S \subseteq [n-1]$: 
$\widehat{f}(S) = \frac{1}{2}(\widehat{g}(S) + \widehat{h}(S))$, and 
$\widehat{f}(S \cup \{n\}) = \frac{1}{2}(\widehat{g}(S) - \widehat{h}(S))$.
We also note that for any two numbers, $|a+b| + |a-b| = 2 \max(|a|,|b|)$.
So we get:
\[| \widehat{f}(S) | + | \widehat{f}(S \cup \{n\}) | = \frac{1}{2} (| \widehat{g}(S) + \widehat{h}(S) | + | \widehat{g}(S) - \widehat{h}(S) |) 
= \frac{1}{2} \cdot 2 \cdot \max(| \widehat{g}(S) |, | \widehat{h}(S) | )  \]
Summing over all $S \subseteq [n-1]$, we get:
\[\sum_{S \subseteq [n]}{|\widehat{f}(S)|} = \sum_{S \subseteq [n-1]}{|\widehat{g}(S)|} + \sum_{S \subseteq [n-1]}{|\widehat{h}(S)|}
- \sum_{S \subseteq [n-1]}{\min(|\widehat{g}(S)|,|\widehat{h}(S)|) }\]
In other words,
\[ \|\widehat{f}\|_1 = \|\widehat{g}\|_1 + \|\widehat{h}\|_1 - \sum_{S \subseteq [n]}{\min(|\widehat{g}(S)|,|\widehat{h}(S)|)} \] 
We can bound the last term as follows:
\[\sum_{S \subseteq [n]}{\min(|\widehat{g}(S)|,|\widehat{h}(S)|)} \geq \sum_{S \subseteq [n]}{|\widehat{g}(S)\widehat{h}(S)|}
\geq \sum_{\emptyset \neq S \subseteq [n]}{|\widehat{g}(S)\widehat{h}(S)|}
\geq |\sum_{\emptyset \neq S \subseteq [n]}{\widehat{g}(S)\widehat{h}(S)}| = | \Cov(g,h) | \]
So in conclusion, we get:
\begin{equation}
\label{L1 induction}
\widehat{f}\|_1 \leq \|\widehat{g}\|_1 + \|\widehat{h}\|_1 - | \Cov(g,h) |
\end{equation}

We can now prove $\|\widehat{f}\|_1 \leq \text{boundary\_size}(T) - \sum_{v \in \text{inner}(T)}{| \Cov(g_v, h_v) |}$ using structural induction on $T$:
\begin{itemize}
\item The \textbf{base case} is where we have a function with a root and two leaves. $\|\widehat{f}\|_1 = 1$, $\text{boundary\_size}(T) = 1$, and there are no inner nodes so the claim holds.

\item The \textbf{semi induction step} is where we have a function with a root, one leaf child and one non-leaf sub-tree $T'$ computing a function $h$. In this case, $\|\widehat{f}\|_1 = 1 + \|\widehat{h}\|_1$, $\text{boundary\_size}(T) = 1 + \text{boundary\_size}(T')$, and $\sum_{v \in \text{inner}(T)}{| \Cov(g_v, h_v) |} = \sum_{v \in \text{inner}(T')}{| \Cov(g_v, h_v) |}$, as the inner nodes in $T$ and $T'$ are the same. Using the inductive hypothesis on $T'$ is enough to finish this case.

\item The \textbf{induction step} is where we have a function with a root and two non-leaf sub-trees $T', T''$ computing $g,h$. We use the inductive hypothesis for the two sub-trees and inequality \ref{L1 induction}:

\begin{align*}
\|\widehat{f}\|_1 &\leq \|\widehat{g}\|_1 + \|\widehat{h}\|_1 - | \Cov(g,h) |  \\
& \leq \text{boundary\_size}(T') - \sum_{v \in \text{inner}(T')}{| \Cov(g_v, h_v) |} \\
& + \text{boundary\_size}(T'') - \sum_{v \in \text{inner}(T'')}{| \Cov(g_v, h_v) |} - | \Cov(g,h) | \\
& = \text{boundary\_size}(T) - \sum_{v \in \text{inner}(T)}{| \Cov(g_v, h_v) |} 
\end{align*}

\end{itemize}
These three are the only possible cases so we are done.

\end{proof}

\addcontentsline{toc}{section}{\protect\numberline{}Acknowledgements}%
\section*{Acknowledgements}
The author wishes to thank Amir Shpilka for advising him throughout this research, and also Dor Minzer and Ben Lee Volk for fruitful discussions and suggestions.

\addcontentsline{toc}{section}{\protect\numberline{}References}%
\bibliographystyle{alpha}
\bibliography{bibliography}

\end{document}